\documentclass[11pt]{article}

\usepackage[T1]{fontenc}
\usepackage{textcomp}
\usepackage{palatino}
\usepackage{mathpazo}
\usepackage{stmaryrd}

\usepackage{hyperref}
\hypersetup{pdfpagemode=UseNone}

\usepackage{amsfonts}
\usepackage{amssymb}
\usepackage{amsmath}
\usepackage{latexsym}
\usepackage{amsthm}
\usepackage{eepic}
\usepackage{sectsty}
\usepackage[usenames]{color}

\newtheorem{theorem}{Theorem}
\newtheorem{lemma}[theorem]{Lemma}

\theoremstyle{definition}

\newcommand{\tinyspace}{\mspace{1mu}}
\newcommand{\microspace}{\mspace{0.5mu}}

\newcommand{\norm}[1]{\left\lVert\tinyspace#1\tinyspace\right\rVert}
\newcommand{\snorm}[1]{\lVert\tinyspace#1\tinyspace\rVert}

\newcommand{\tr}{\operatorname{Tr}}

\newcommand{\ip}[2]{\left\langle #1 , #2\right\rangle}

\def\({\left(}
\def\){\right)}
\def\I{\mathbb{1}}

\newcommand{\setft}[1]{\mathrm{#1}}
\newcommand{\lin}[1]{\setft{L}\left(#1\right)}
\newcommand{\density}[1]{\setft{D}\left(#1\right)}

\newcommand{\pos}[1]{\setft{Pos}\left(#1\right)}

\def\complex{\mathbb{C}}

\def \lket {\left|}
\def \rket {\right\rangle}
\def \lbra {\left\langle}
\def \rbra {\right|}
\newcommand{\ket}[1]{\lket\microspace #1 \microspace\rket}
\newcommand{\bra}[1]{\lbra\microspace #1 \microspace\rbra}

\newenvironment{mylist}[1]{\begin{list}{}{
	\setlength{\leftmargin}{#1}
	\setlength{\rightmargin}{0mm}
	\setlength{\labelsep}{2mm}
	\setlength{\labelwidth}{8mm}
	\setlength{\itemsep}{0mm}}}
	{\end{list}}

\newcommand{\class}[1]{\textup{#1}}

\newcommand{\reg}[1]{\mathsf{#1}}

\def\X{\mathcal{X}}
\def\Y{\mathcal{Y}}
\def\Z{\mathcal{Z}}
\def\W{\mathcal{W}}

\def\V{\mathcal{V}}

\def\D{\mathcal{D}}

\def\P{\mathcal{P}}

\def\yes{\text{yes}}
\def\no{\text{no}}

\begin{document}

%-----------------------------------------------------------------------------%
\title{\bf QIP = PSPACE}
%-----------------------------------------------------------------------------%

\author{%
  Rahul Jain${}^\ast$
  \quad\quad
  Zhengfeng Ji${}^\dagger$
  \quad\quad
  Sarvagya Upadhyay${}^\ddagger$
  \quad\quad
  John Watrous${}^\ddagger$\\[4mm]
  ${}^\ast$%
  {\small\it Department of Computer Science and Centre for
    Quantum Technologies}\\[-1mm]
  {\small\it National University of Singapore}\\[-1mm]
  {\small\it Republic of Singapore}\\[2mm]
  ${}^\dagger$%
  {\small\it Perimeter Institute for Theoretical Physics}\\[-1mm]
  {\small\it Waterloo, Ontario, Canada}\\[2mm]
  ${}^\ddagger$%
  {\small\it Institute for Quantum Computing and School of Computer
  Science}\\[-1mm]
  {\small\it University of Waterloo}\\[-1mm]
  {\small\it Waterloo, Ontario, Canada}
}

\date{August 3, 2009}

\maketitle

\begin{abstract}
  We prove that the complexity class QIP, which consists of all
  problems having quantum interactive proof systems, is contained in
  PSPACE.
  This containment is proved by applying a parallelized form of the
  matrix multiplicative weights update method to a class of
  semidefinite programs that captures the computational power of
  quantum interactive proofs.
  As the containment of PSPACE in QIP follows immediately from the
  well-known equality IP = PSPACE, the equality QIP = PSPACE follows.
\end{abstract}

%-----------------------------------------------------------------------------%
\section{Introduction} \label{sec:introduction}
%-----------------------------------------------------------------------------%

Efficient proof verification is a fundamental notion in computational
complexity theory.
The most direct complexity-theoretic abstraction of efficient proof
verification is represented by the complexity class $\class{NP}$,
wherein a deterministic polynomial-time {\it verification procedure}
decides whether a given polynomial-length {\it proof string} is valid
for a given input.
One cannot overstate the importance of this class and its presently
unknown relationship to $\class{P}$, the class of problems solvable
deterministically in polynomial time.
This problem, which is known as the $\class{P}$ versus $\class{NP}$
problem, is one of the greatest of all unsolved problems in
mathematics.

In the early to mid 1980's, Babai \cite{Babai85} and Goldwasser,
Micali, and Rackoff \cite{GoldwasserMR85} introduced a computational
model that extends the notion of efficient proof verification to 
{\it interactive settings}. 
(Journal versions of these papers appeared later as
\cite{BabaiM88} and \cite{GoldwasserMR89}.)
In this model, which is known as the {\it interactive proof system}
model, a computationally bounded {\it verifier} interacts with a 
{\it prover} of unlimited computation power.
The interaction comprises one or more rounds of communication between
the prover and verifier, and the verifier may make use of randomly
generated bits during the interaction.
After the rounds of communication are finished, the verifier makes a
decision to {\it accept} or {\it reject} based on the interaction.

A decision problem $A$ is said to have an interactive proof system if
there exists a verifier, always assumed to run in polynomial time,
that meets two conditions: the {\it completeness} condition and the
{\it soundness} condition.
The completeness condition formalizes the requirement that true
statements can be proved, which in the present setting means that if
an input string $x$ is a yes-instance of $A$, then there exists a
course of action for the prover that causes the verifier to accept
with high probability.
The soundness condition formalizes the requirement that false
statements cannot be proved, meaning in this case that if an input
string $x$ is a no-instance of $A$, then the verifier will reject with
high probability no matter what course of action the prover takes.
One denotes by $\class{IP}$ the collection of decision
problems having interactive proof systems.
(Here, and throughout the rest of the paper, we take the term
{\it problem} to mean {\it promise problem}, and consider that all
complexity classes to be discussed are classes of promise problems.
Promise problems were defined by Even, Selman and Yacobi
\cite{EvenSY84}, and readers unfamiliar with them are referred to the
survey of Goldreich \cite{Goldreich05}.)

The expressive power of interactive proof systems was not initially
known when they were first defined, but it was soon determined to
coincide with $\class{PSPACE}$, the class of problems solvable
deterministically in polynomial space.
The containment $\class{IP}\subseteq\class{PSPACE}$, which is
generally attributed to Feldman \cite{Feldman86}, is fairly
straightforward---and readers not interested in proving this fact for
themselves can find a proof in \cite{HemaspaandraO02}.
Known proofs \cite{LundFKN92,Shamir92,Shen92} of the reverse
containment $\class{PSPACE}\subseteq\class{IP}$, on the other hand,
are not straightforward, and make essential use of a technique
commonly known as {\it arithmetization}.
This technique involves the extension of Boolean formulas to
multivariate polynomials over large finite fields whose 0 and 1
elements are taken to represent Boolean values.
Through the use of randomness and polynomial interpolation, verifiers
may be constructed for arbitrary $\class{PSPACE}$ problems.

Many variants of interactive proof systems have been studied,
including public-coin interactive proofs
\cite{Babai85,BabaiM88,GoldwasserS89}, multi-prover interactive proofs
\cite{Ben-OrGKW88}, zero-knowledge interactive proofs
\cite{GoldwasserMR89,GoldreichMW91}, and competing-prover interactive
proofs \cite{FeigeK97}.
The present paper is concerned with
{\it quantum interactive proof systems}, which were first studied a
decade after $\class{IP}=\class{PSPACE}$ was proved
\cite{Watrous99-qip-focs, KitaevW00}.
The fundamental notions of this model are the same as those of
classical interactive proof systems, except that the prover and
verifier may now process and exchange quantum information.
Similar to the classical case, several variants of quantum
interactive proof systems have been studied (including those
considered in 
\cite{HallgrenKSZ08,KempeKMV09,KobayashiM03,Kobayashi08,MarriottW05,
  Watrous09}).

One of the most interesting aspects of quantum interactive proof
systems, which distinguishes them from classical interactive proof
systems (at least to the best of our current knowledge), is that they
can be {\it parallelized} to three messages.
That is, quantum interactive proof systems consisting of just three
messages exchanged between the prover and verifier already have the
full power of quantum interactive proofs having a polynomial number of
messages \cite{KitaevW00}.
Classical interactive proofs are not known to hold this property, and
if they do the polynomial-time hierarchy collapses to the second level
\cite{BabaiM88}.

The complexity class $\class{QIP}$ is defined as the class of decision
problems having quantum interactive proof systems.
$\class{QIP}$ trivially contains $\class{IP}$, as the ability of a
verifier to process quantum information is never a hindrance---a
quantum verifier can simulate a classical verifier, and a
computationally unbounded prover can never use quantum information
to an advantage against a verifier behaving classically.
The inclusion $\class{PSPACE}\subseteq\class{QIP}$ is therefore
immediate.
The best upper bound on $\class{QIP}$ known prior to the present paper
was $\class{QIP} \subseteq \class{EXP}$, which was proved in
\cite{KitaevW00} through the use of semidefinite programming.
The optimal probability with which a given verifier can be made to
accept in a quantum interactive proof system can be represented as an
exponential-size semidefinite program, and known polynomial-time
algorithms for semidefinite programming provide the required tool to
prove the containment.
It has been an open problem for the last decade to establish more
precise bounds on the class $\class{QIP}$.

It was recently shown in the paper \cite{JainUW09} that
$\class{QIP}(2)$, the class of problem having 2-message quantum
interactive proof systems, is contained in $\class{PSPACE}$.
That paper made use of a parallel algorithm, based on a method known
as the {\it matrix multiplicative weights update method}, to
approximate optimal solutions for a class of semidefinite programs
that represent the maximum acceptance probabilities for verifiers in
two-message quantum interactive proofs.
In this paper we extend this result to all of $\class{QIP}$,
establishing the relationship $\class{QIP} = \class{PSPACE}$.
Similar to \cite{JainUW09}, we use the matrix multiplicative weights
update method, together with parallel methods for matrix computations.

The {\it multiplicative weights method} is a framework for algorithm
design having its origins in various fields, including learning
theory, game theory, and optimization. 
Its matrix variant, as discussed in the survey paper \cite{AroraHK05}
and the PhD thesis of Kale \cite{Kale07}, gives an iterative way to
approximate the optimal value of semidefinite programs
\cite{AroraK07,WarmuthK06}. 
In addition to its application in \cite{JainUW09}, it was applied to
quantum complexity in \cite{JainW09} to prove the containment of
the complexity class $\class{QRG(1)}$ in $\class{PSPACE}$.
The key strength of this method for these applications is that it can
be parallelized for some special classes of semidefinite programs.

A key result that allows our technique to work for the
entire class $\class{QIP}$ is the characterization 
$\class{QIP} = \class{QMAM}$ proved in \cite{MarriottW05}.
This characterization, which is described in greater detail in the
next section, concerns a restricted notion of interactive proof
systems known as {\it Arthur--Merlin games}.
An Arthur--Merlin game is an interactive proof system wherein the
verifier can only send uniformly generated random bits to the prover.
Following Babai \cite{Babai85}, one refers to the verifier
as {\it Arthur} and to the prover as {\it Merlin} in this setting.
It is also typical to refer to the individual bits of Arthur's
messages as {\it coins}, given that they are each uniformly generated
like the flip of a fair coin.
The restriction that Arthur sends only uniformly generated bits to
Merlin, and therefore does not have the option to base his messages on
private information unknown to Merlin, would seem to limit the power
of Arthur--Merlin games in comparison to ordinary interactive proof
systems.
But in fact this is known not to be the case, both for classical
\cite{GoldwasserS89} and quantum \cite{MarriottW05}
interactive proof systems.
In the quantum setting, this characterization admits a significant
simplification in the semidefinite programs that capture
the complexity of the class $\class{QIP}$.

The remainder of this paper has the following organization.
Section~\ref{sec:preliminaries} includes background information,
notation, and other preliminary discussions that are relevant to the
remainder of the paper.
Section~\ref{sec:sdp} describes a semidefinite programming problem
that captures the complexity of the class $\class{QIP}$ based on
quantum Arthur--Merlin games, and Section~\ref{sec:algorithm} presents
the main algorithm that solves this problem.
Finally, Section~\ref{sec:proof} discusses a parallel approximation to
the algorithm from Section~\ref{sec:algorithm} and explains how its
properties lead to the containment $\class{QIP}\subseteq\class{PSPACE}$.

%-----------------------------------------------------------------------------%
\section{Preliminaries} \label{sec:preliminaries}
%-----------------------------------------------------------------------------%

This section contains a summary of the notation and terminology on
linear algebra, quantum information, semidefinite programming,
quantum Arthur--Merlin games, and bounded-depth circuits that is used
later in the paper.
For the most part, these discussions are intended only to make clear
the notation and terminology that we use, and not to provide
introductions to these topics.
We assume that the reader already has familiarity with complexity
theory and quantum computing, and refer readers who are not to
\cite{AroraB09} and \cite{NielsenC00}.

%-----------------------------------------------------------------------------%
\subsection{Linear algebra and quantum information} 
\label{sec:linear-algebra}
%-----------------------------------------------------------------------------%

A {\it quantum register} refers to a collection of qubits, or more
generally a finite-size component in a quantum computer.
Every quantum register $\reg{V}$ has associated with it a finite,
non-empty set $\Sigma$ of classical states and a complex vector space
of the form $\V = \complex^{\Sigma}$.
We use the Dirac notation $\{\ket{a}\,:\,a\in\Sigma\}$ to refer to the
{\it standard basis} (or elementary unit vectors) in $\V$, and define
the inner product and Euclidean norm on $\V$ in the standard way.
The set $\{\bra{a}\,:\,a\in\Sigma\}$ consists of the elements in the
dual space of $\V$ that are in correspondence with the standard basis
vectors.

For such a space $\V$, we write $\lin{\V}$ to denote the space of
linear mappings, or {\it operators}, from $\V$ to itself, which is
identified with the set of square complex matrices indexed by $\Sigma$
in usual way.
An inner product on $\lin{\V}$ is defined as 
\[
\ip{A}{B} = \tr(A^{\ast}B),
\]
where $A^{\ast}$ denotes the adjoint (or conjugate transpose) of $A$.
The identity operator on $\V$ is denoted $\I_{\V}$ (or just $\I$ when
$\V$ is understood).

The following special types of operators are relevant to the paper:
\begin{mylist}{\parindent}
\item[1.]
An operator $A\in\lin{\V}$ is {\it Hermitian} if $A = A^{\ast}$.
The eigenvalues of a Hermitian operator are always real, and for $m =
\dim(\V)$ we write
\[
\lambda_1(A) \geq \lambda_2(A) \geq \cdots \geq \lambda_m(A)
\]
to denote the eigenvalues of $A$ sorted from largest to smallest.

\item[2.]
An operator $P\in\lin{\V}$ is {\it positive semidefinite} if it is
Hermitian and all of its eigenvalues are nonnegative.
The set of such operators is denoted $\pos{\V}$.
The notation $P\geq 0$ also indicates that $P$ is positive
semidefinite, and more generally the notations $A\leq B$ and $B\geq A$
indicate that $B - A\geq 0$ for Hermitian operators $A$ and $B$.

Every Hermitian operator $A$ can be expressed uniquely as $A = P - Q$
for positive semidefinite operators $P$ and $Q$ satisfying
$\ip{P}{Q} = 0$.
The operator $P$ is said to be the {\it positive part} of $A$, while
$Q$ is the {\it negative part}.

\item[3.]
A positive semidefinite operator $P\in\pos{\V}$ is also
said to be {\it positive definite} if all of its eigenvalues are
positive (which implies that $P$ must be invertible).
The notation $P>0$ also indicates that $P$ is positive definite, and
the notations $A<B$ and $B>A$ indicate that $B - A > 0$ for
Hermitian operators $A$ and $B$.

\item[4.]
An operator $\rho\in\pos{\V}$ is a {\it density operator} if it is
both positive semidefinite and has trace equal to 1.
The set of such operators is denoted $\density{\V}$.

\item[5.]
An operator $\Pi\in\pos{\V}$ is a {\it projection} if all of its
eigenvalues are either 0 or 1.

\end{mylist}

A {\it quantum state} of a register $\reg{V}$ is a density operator
$\rho\in\density{\V}$, and a {\it measurement} on $\reg{V}$ is a
collection $\{P_b\,:\,b\in\Gamma\}\subseteq\pos{\V}$ satisfying
\[
\sum_{b\in\Gamma}P_b = \I_{\V}.
\]
The set $\Gamma$ is the set of {\it measurement outcomes}, and when
such a measurement is performed on $\reg{V}$ while it is in the state
$\rho$, each outcome $b\in\Gamma$ occurs with probability
$\ip{P_b}{\rho}$.

The {\it spectral norm} of an operator $A\in\lin{\V}$ is defined as
\[
\norm{A} = \max\{ \norm{Av}\,:\,v\in\V,\,\norm{v} = 1\}.
\]
The spectral norm is sub-multiplicative, meaning that 
$\norm{AB} \leq \norm{A}\norm{B}$ for all operators $A,B\in\lin{\V}$,
and it holds that $\norm{P} = \lambda_1(P)$ for every positive
semidefinite operator $P$.
For any operator $A\in\lin{\V}$, the exponential of $A$ is defined as
\[
\exp(A) = \I + A + A^2/2 + A^3/6 + \cdots
\]
The {\it Golden-Thompson Inequality} (see Section IX.3 of
\cite{Bhatia97}) states that, for any two
Hermitian operators $A$ and $B$ on $\V$, we have
\[
\tr\left[ \exp(A + B)\right] \leq
\tr\left[\exp(A)\exp(B)\right].
\]

The tensor product $\V\otimes\W$ of vector spaces 
$\V = \complex^{\Sigma}$ and $\W = \complex^{\Gamma}$ may be
associated with the space $\complex^{\Sigma\times\Gamma}$, and
the tensor product of operators $A\in\lin{\V}$ and $B\in\lin{\W}$ is
then taken to be the unique operator $A\otimes B \in
\lin{\V\otimes\W}$ satisfying 
$(A\otimes B)(v\otimes w) = (Av) \otimes (Bw)$ for all $v\in\V$ and
$w\in\W$.
These notions may be associated with the usual Kronecker product of
vectors and matrices.
For quantum registers $\reg{V}$ and $\reg{W}$, the space $\V\otimes\W$
is associated with the pair $(\reg{V},\reg{W})$, viewed as a single
register.
Tensor products involving three or more spaces are handled similarly.

For a given linear mapping of the form
$\Phi:\lin{\V}\rightarrow\lin{\W}$, one defines the adjoint mapping
$\Phi^{\ast}:\lin{\W}\rightarrow\lin{\V}$ to be the unique linear
mapping that satisfies
\[
\ip{B}{\Phi(A)} = \ip{\Phi^{\ast}(B)}{A}
\]
for all operators $A\in\lin{V}$ and $B\in\lin{\W}$.

Finally, for spaces $\V$ and $\W$, one defines the {\it partial trace}
$\tr_{\V} : \lin{\V\otimes\W}\rightarrow\lin{\W}$ to be the unique
linear mapping that satisfies
$\tr_{\V}(A\otimes B) = (\tr A) B$ for all $A\in\lin{\V}$ and
$B\in\lin{\W}$.
A similar notation is used for the partial trace $\tr_{\W}$, or
partial traces defined on three or more tensor factors.
When this notation is used, the spaces on which the trace is not taken
are determined by context.
When a pair of registers $(\reg{V},\reg{W})$ is viewed as a single
register and has the quantum state $\rho\in\density{\V\otimes\W}$, one
defines the state of $\reg{W}$ to be $\tr_{\V}(\rho)$.
In other words, the partial trace describes the action of destroying,
or simply ignoring, a given quantum register.

%-----------------------------------------------------------------------------%
\subsection{Semidefinite programming} \label{sec:sdp-prelim}
%-----------------------------------------------------------------------------%

A {\it semidefinite program} over complex vector spaces $\V$ and
$\W$ is a pair of optimization problems as follows.

\begin{center}
  \begin{minipage}{2in}
    \centerline{\underline{Primal problem}}\vspace{-7mm}
    \begin{align*}
      \text{maximize:}\quad & \ip{C}{X}\\
      \text{subject to:}\quad & \Psi(X) \leq D,\\
      & X\in\pos{\V}.
    \end{align*}
  \end{minipage}
  \hspace*{25mm}
  \begin{minipage}{2in}
    \centerline{\underline{Dual problem}}\vspace{-7mm}
    \begin{align*}
      \text{minimize:}\quad & \ip{D}{Y}\\
      \text{subject to:}\quad & \Psi^{\ast}(Y) \geq C,\\
      & Y\in\pos{\W}.
    \end{align*}
  \end{minipage}
\end{center}

\noindent
Here, the operators $C\in\lin{\V}$ and $D\in\lin{\W}$ are Hermitian
and $\Psi:\lin{\V}\rightarrow\lin{\W}$ must be a linear mapping that 
maps Hermitian operators to Hermitian operators.
Readers familiar with semidefinite programming will note that the
above form of a semidefinite program is different from the
well-known {\it standard form}, but it is equivalent and better suited
for this paper's needs.
The form given above is, in essence, the one that is typically
followed for general conic programming \cite{BoydV04}.

It is typical that semidefinite programs are stated in forms that do
not explicitly describe $\Psi$, $C$ and $D$, and the same is true for
the semidefinite programs we will consider.
It is, however, routine to put them into the above form.

With the above optimization problems in mind, one defines the 
{\it primal feasible} set $\P$ and the {\it dual feasible} set $\D$ as
\begin{align*}
\P & = \left\{X\in\pos{\V}\,:\,\Psi(X) \leq D\right\},\\
\D & = \left\{Y\in\pos{\W}\,:\,\Psi^{\ast}(Y) \geq C\right\}.
\end{align*}
Operators $X\in\P$ and $Y\in\D$ are also said to be {\it primal feasible}
and {\it dual feasible}, respectively.
The functions $X\mapsto\ip{C}{X}$ and $Y\mapsto\ip{D}{Y}$ are
called the primal and dual {\it objective functions}, and the
{\it optimal values} associated with the primal and dual problems
are defined as
\[
\alpha = \sup_{X\in\P} \ip{C}{X}\quad\quad\text{and}\quad\quad
\beta = \inf_{Y\in\D} \ip{D}{Y}.
\]

Semidefinite programs have associated with them a powerful theory of
{\it duality}, which refers to the special relationship between the
primal and dual problems.
The property of {\it weak duality}, which holds for all semidefinite
programs, states that $\alpha \leq \beta$.
This property implies that every dual feasible operator
$Y\in\D$ provides an upper bound of $\ip{D}{Y}$ on the value
$\ip{C}{X}$ that is achievable over all choices of a primal feasible
$X\in\P$, and likewise every primal feasible operator $X\in\P$
provides a lower bound of $\ip{C}{X}$ on the value $\ip{D}{Y}$ that
is achievable over all choices of a dual feasible $Y\in\D$.

It is not always the case that $\alpha = \beta$ for a given
semidefinite program, but in most natural cases it does hold.
The situation in which $\alpha = \beta$ is known as 
{\it strong duality}, and several conditions have been identified that
imply strong duality.
One such condition is {\it strict dual feasibility}:
if $\alpha$ is finite and there exists an operator $Y>0$ such that
$\Psi^{\ast}(Y) > C$, then $\alpha = \beta$.
The symmetric condition of {\it strict primal feasibility} also
implies strong duality.

%-----------------------------------------------------------------------------%
\subsection{Single-coin quantum Arthur--Merlin games} \label{sec:qmam}
%-----------------------------------------------------------------------------%

Quantum Arthur--Merlin games were proposed in \cite{MarriottW05} as a
natural quantum variant of classical Arthur--Merlin games.
Here, one simply mimics the classical definition in requiring that
Arthur's messages to Merlin consist of uniformly generated random
bits.
Merlin's messages to Arthur, however, may be quantum; and after all of
the messages have been exchanged Arthur is free to perform a quantum
computation when deciding to accept or reject.

Of particular interest to us are quantum Arthur--Merlin games in which
three messages are exchanged, and where Arthur's only message consists
of a single bit.
In more precise terms, such an interaction takes the following form:
\begin{mylist}{\parindent}
\item[1.]
Merlin sends a quantum register $\reg{W}$ to Arthur.
Merlin is free to initialize this register to any quantum state of his
choice, and may entangle it with a register of his own if he chooses.

\item[2.]
After receiving $\reg{W}$ from Merlin, Arthur chooses a bit
$a\in\{0,1\}$ uniformly at random.
Merlin learns the value of $a$.

\item[3.]
Merlin sends Arthur a second quantum register $\reg{Y}$.
He does this after step 2, so he has the option to condition the state
of $\reg{Y}$ upon the value of $a$.
The register $\reg{Y}$ could, of course, be entangled with $\reg{W}$
in any way that quantum information theory permits.

\item[4.]
After receiving $\reg{Y}$, Arthur performs one of two binary-valued
measurements, determined by the value of the random bit $a$, on the
pair $(\reg{W},\reg{Y})$.
The measurement outcome 1 is interpreted as {\it acceptance}, while 0
is interpreted as {\it rejection}.
\end{mylist}

\noindent
Arthur's measurements must of course be efficiently implementable.
This notion is formalized by requiring that the measurements are
implementable by polynomial-time generated families of quantum
circuits, which naturally requires the registers $\reg{W}$ and
$\reg{Y}$ to consist of a number of qubits that is polynomial in the
length of the input.
Further details may be found in \cite{MarriottW05}.

The result of \cite{MarriottW05} that we make use of is that every
problem $A \in \class{QIP}$ has a single-coin Arthur--Merlin game as just
described.
The game is such that if $x$ is a yes-instance of the problem $A$,
then Arthur accepts with probability 1, whereas if the input $x$ is a
no-instance of the problem then Arthur accepts with probability at
most $1/2 + \varepsilon$, for any desired constant $\varepsilon>0$.
(In the construction given in \cite{MarriottW05}, Arthur's
measurements are always nontrivial projective measurements.
This implies that even for no-instance inputs, Merlin can cause Arthur to
accept with probability at least 1/2 by simply guessing in advance
Arthur's random bit.)

%-----------------------------------------------------------------------------%
\subsection{Bounded-depth circuit complexity} 
\label{sec:NC}
%-----------------------------------------------------------------------------%

In the last section of the paper, we will require the definitions of
two complexity classes based on bounded-depth circuit families:
$\class{NC}$ and $\class{NC}(\mathit{poly})$.
It is convenient for us to define these as classes of {\it functions}
rather than decision problems, and when we wish to view them as
classes of decision problems we simply restrict our attention to 
binary-valued functions.
The class $\class{NC}$ contains all functions computable by
logarithmic-space uniform Boolean circuits of polylogarthmic depth,
and $\class{NC}(\mathit{poly})$ contains all functions that can be
computed by polynomial-space uniform families of Boolean circuits
having polynomial-depth.
For decision problems it is known \cite{Borodin77} that
$\class{NC}(\mathit{poly}) = \class{PSPACE}$, and the proof of our
main result will make use of this fact.

There are two fundamental properties of $\class{NC}(\mathit{poly})$
that we will take advantage of.
The first is that functions in
$\class{NC}$ and $\class{NC}(\mathit{poly})$ compose 
well, and the second is that many computational problems involving
matrices are in $\class{NC}$.
In more precise terms, the first property is as follows.
If $F:\{0,1\}^{\ast} \rightarrow \{0,1\}^{\ast}$ is a function in
$\class{NC}(\mathit{poly})$ and $G:\{0,1\}^{\ast}\rightarrow
\{0,1\}^{\ast}$ is a function in $\class{NC}$, then the composition
$G\circ F$ is also in $\class{NC}(\mathit{poly})$.
This follows from the most straightforward way of composing the
families of circuits that compute $F$ and $G$.

To discuss the second property, it will be helpful to make clear our
assumptions concerning matrix computations.
We will always assume that the matrices on which computations are
performed have entries with rational real and imaginary parts, and
that the rational numbers are represented as pairs of integers in
binary notation.
Unless it is explicitly noted otherwise, any other rational numbers
involved in our computations will be represented in a similar way.

With these assumptions in place, we first note that elementary matrix
operations, including inverses and iterated sums and products of
matrices, are known to be in $\class{NC}$.
There is an extensive literature on this topic, and we refer the
reader to von zur Gathen's survey \cite{vzGathen93} for more details.
We also note that matrix exponentials and spectral decompositions can
be {\it approximated} to high accuracy in $\class{NC}$.
In more precise terms, the following two problems are in
$\class{NC}$.

\pagebreak[3]

\begin{center}
\underline{Matrix exponentials}\\[2mm]
\begin{tabular}{lp{5.5in}}
{\it Input:} &
An $n\times n$ matrix $M$, a positive rational number $\eta$,
and an integer $k$ expressed in unary notation (i.e., $1^k$).\\[1mm]
{\it Promise:} &
$\norm{M} \leq k$.\\[1mm]
{\it Output:} &
An $n\times n$ matrix $X$ such that 
$\norm{\exp(M) - X} < \eta$.
\end{tabular}
\\[6mm]
\underline{Spectral decompositions}\\[2mm]
\begin{tabular}{lp{5.5in}}
{\it Input:} &
An $n\times n$ Hermitian matrix $H$ and a positive rational number
$\eta$.\\[1mm]
{\it Output:} &
An $n\times n$ unitary matrix $U$ and an $n\times n$ real diagonal matrix
$\Lambda$ such that
\[
\norm{M - U \Lambda U^{\ast}} < \eta.
\]
\end{tabular}
\end{center}

\vspace{-6mm}

\noindent
The reader will note that in these problems, the description of the
error parameter $\eta$ could require as few as $O(\log(1/\eta))$ bits.
This implies that highly accurate approximations, for instance where
$\eta = 2^{-n}$, are possible in $\class{NC}$.
The fact that matrix exponentials can be approximated in $\class{NC}$
follows by truncating the series 
\[
\exp(M) = \I + M + M^2/2 + M^3/6 + \cdots
\]
to a number of terms linear in $k+\log(1/\eta)$.
(From a numerical point of view this is not a very good way to compute
matrix exponentials \cite{MolerV03}, but it is arguably the simplest
way to prove that the stated problem is in $\class{NC}$.)
The fact that spectral decompositions can be approximated in
$\class{NC}$ follows from a composition of known facts: in
$\class{NC}$ one can compute characteristic polynomials and null
spaces of matrices, perform orthogonalizations of vectors, and
approximate roots of integer polynomials to high precision
\cite{Csanky76,BorodinGH82,BorodinCP83,BenOrFKT86,vzGathen93,Neff94}.

%-----------------------------------------------------------------------------%
\section{A semidefinite programming formulation of the problem}
\label{sec:sdp}
%-----------------------------------------------------------------------------%

Consider Arthur's verification procedure for a given single-coin QMAM
protocol on a fixed input string $x$.
Arthur first receives a register $\reg{W}$, then generates a random
bit $a\in\{0,1\}$, and then receives a second register $\reg{Y}$.
He then measures $(\reg{W},\reg{Y})$ with respect to a
binary-valued measurement
\[
\{P_a,\I - P_a\} \subset \pos{\W\otimes\Y},
\]
where we take each of the operators $P_0$ and $P_1$ to represent
acceptance and $\I - P_0$ and $\I - P_1$ to represent rejection.
If the quantum state of $(\reg{W},\reg{Y})$ is given by a density
operator $\rho\in\density{\W\otimes\Y}$ when Arthur measures,
he will therefore accept with probability $\ip{P_a}{\rho}$.

Now define
\[
Q = \frac{1}{2} \ket{0}\bra{0} \otimes P_0 +
\frac{1}{2} \ket{1}\bra{1} \otimes P_1 \in
\pos{\X\otimes\W\otimes\Y},
\]
where we take $\X = \complex^{\{0,1\}}$ to be the vector space
corresponding to Arthur's random choice of $a\in\{0,1\}$, and consider
the optimal probability that Merlin can cause Arthur to accept.
If, for each of the values $a\in\{0,1\}$, Merlin is able to leave the
state $\rho_a$ in the registers $(\reg{W},\reg{Y})$ right before
Arthur measures, he will convince Arthur to accept with probability
\begin{equation} \label{eq:Ps-and-Q}
\frac{1}{2} \ip{P_0}{\rho_0} + \frac{1}{2} \ip{P_1}{\rho_1}
=
\ip{Q}{X}
\end{equation}
for
\[
X = \ket{0}\bra{0} \otimes \rho_0 + \ket{1}\bra{1} \otimes \rho_1.
\]
There is, of course, a constraint on Merlin's choice of $\rho_0$ and
$\rho_1$, which is that they must agree on $\reg{W}$, as Merlin cannot
touch the register $\reg{W}$ at any point after Arthur chooses the
random bit $a$.
In more precise terms, it must hold that
\begin{equation} \label{eq:agree-on-W}
\tr_{\Y}(\rho_0) = \sigma = \tr_{\Y}(\rho_1)
\end{equation}
for some density operator $\sigma\in\density{\W}$.
This, in fact, is Merlin's only constraint---for if he holds a
purification of the state $\sigma$, he is free to set the state of
$(\reg{W},\reg{Y})$ to any choice of $\rho_0$ and $\rho_1$ satisfying
\eqref{eq:agree-on-W} without needing access to $\reg{W}$.

Now, we note that the condition \eqref{eq:agree-on-W} implies that
\begin{equation} \label{eq:condition-on-X}
\tr_{\Y}(X) = \I_{\X} \otimes \sigma.
\end{equation}
Moreover, for an arbitrary operator $X\in\pos{\X\otimes\W\otimes\Y}$
satisfying the constraint \eqref{eq:condition-on-X}, one has that the
operators $\rho_0$ and $\rho_1$ defined as
\[
\rho_a = 
\(\bra{a}\otimes\I_{\W\otimes\Y}\)X\(\ket{a}\otimes\I_{\W\otimes\Y}\)
\]
for $a\in\{0,1\}$ satisfy the conditions \eqref{eq:Ps-and-Q} and
\eqref{eq:agree-on-W}.
It follows that the following semidefinite program represents the
optimal probability with which Merlin can convince Arthur to accept.

\begin{center}
  \begin{minipage}[t]{3in}
    \centerline{\underline{Primal problem}}\vspace{-7mm}
    \begin{align*}
      \text{maximize:}\quad & \ip{Q}{X}\\
      \text{subject to:}\quad & \tr_{\Y}(X) \leq \I_{\X}\otimes\sigma,\\
      & X\in\pos{\X\otimes\W\otimes\Y},\\
      & \sigma\in\density{\W}.
    \end{align*}
  \end{minipage}
  \begin{minipage}[t]{3in}
    \centerline{\underline{Dual problem}}\vspace{-7mm}
    \begin{align*}
      \text{minimize:}\quad & \norm{\tr_{\X}(Y)}\\
      \text{subject to:}\quad & Y\otimes\I_{\Y} \geq Q,\\
      & Y\in\pos{\X\otimes\W}.
    \end{align*}
  \end{minipage}
\end{center}

\noindent
Note that the inequality in the primal problem can be exchanged for an
equality without changing the optimal value.
This is because any primal feasible $X$ can be inflated to achieve the
equality $\tr_{\Y}(X) = \I_{\X}\otimes\sigma$ for some choice of
$\sigma$, and this can only increase the value of the objective
function by virtue of the fact that $Q$ is positive semidefinite.
It is immediate that the optimal solution to the primal problem is
bounded and the dual problem is strictly feasible, from which strong
duality follows; the primal and dual problems have the same optimal
values.

Now, under the assumption that $Q$ is invertible, one may perform a
change of variables to put the above semidefinite program into a form
that more closely resembles the one in \cite{JainUW09}.
To do this we define a linear mapping
$\Phi:\lin{\X\otimes\W\otimes\Y}\rightarrow\lin{\X\otimes\W}$ as
\begin{equation} \label{eq:definition-of-Phi}
\Phi(X) = \tr_{\Y}\(Q^{-1/2} X Q^{-1/2}\),
\end{equation}
whose adjoint mapping
$\Phi^{\ast}:\lin{\X\otimes\W}\rightarrow\lin{\X\otimes\W\otimes\Y}$
is given by
\[
\Phi^{\ast}(Y) = Q^{-1/2}(Y\otimes\I_{\Y})Q^{-1/2},
\]
and consider the following semidefinite program.

\begin{center}
  \begin{minipage}[t]{3in}
    \centerline{\underline{Primal problem}}\vspace{-7mm}
    \begin{align*}
      \text{maximize:}\quad & \tr(X)\\
      \text{subject to:}\quad & \Phi(X) \leq \I_{\X}\otimes\sigma,\\
      & X\in\pos{\X\otimes\W\otimes\Y},\\
      & \sigma\in\density{\W}.
    \end{align*}
  \end{minipage}
  \begin{minipage}[t]{3in}
    \centerline{\underline{Dual problem}}\vspace{-7mm}
    \begin{align*}
      \text{minimize:}\quad & \norm{\tr_{\X}(Y)}\\
      \text{subject to:}\quad & \Phi^{\ast}(Y) \geq
      \I_{\X\otimes\W\otimes\Y},\\
      & Y\in\pos{\X\otimes\W}.
    \end{align*}
  \end{minipage}
\end{center}

\noindent
It is clear that this semidefinite program has the same optimal value
as the previous one.

We will be interested in the optimal value of this semidefinite
program in the case that $\snorm{Q^{-1}}$ is upper-bounded by a fixed
constant and where there is a promise on the optimal value.
The promise, which will come from the properties of the quantum
Arthur--Merlin games under consideration, is that the optimal value
does not lie in the interval $(5/8,\,7/8)$, and the goal is to
determine whether the optimal value is larger than $7/8$ or smaller
than $5/8$.

For readers familiar with the semidefinite program for
$\class{QIP}(2)$ presented in \cite{JainUW09}, we note that there are
two essential differences between it and the one above.
The first difference is that the semidefinite program in
\cite{JainUW09} effectively replaces the density operator $\sigma$
with the scalar value 1, which would seem to suggest added difficulty
for the case at hand.
The second difference is that $\X$ is two-dimensional for the
semidefinite program above, whereas it has arbitrary size in 
\cite{JainUW09}.
This second difference more than compensates for the difficulty
induced by the first, and we find that the above semidefinite program
is actually much easier to solve than the one for $\class{QIP}(2)$.

%-----------------------------------------------------------------------------%
\section{The main algorithm and its analysis} \label{sec:algorithm}
%-----------------------------------------------------------------------------%

We now present the main algorithm for the semidefinite programming
problem from the previous section.
The algorithm, which is described in Figure~\ref{fig:algorithm}, takes
as input an operator
\[
Q\in\pos{\X\otimes\W\otimes\Y}.
\]
It is assumed that $Q$ is invertible and satisfies
$\snorm{Q^{-1}}\leq 64$.
(The algorithm could easily be adapted to handle any other fixed
constant in place of 64, but this choice is sufficient for our needs.)
Moreover, it is assumed that the optimal value of the semidefinite
program in Section~\ref{sec:sdp} that is defined by $Q$ does not lie
in the interval $(5/8,\,7/8)$.
Our goal is to prove that the algorithm accepts when the optimal value
is at least $7/8$ and rejects when the optimal value is at most $5/8$.

Here we present the correctness of the algorithm under the assumption
that all computations are performed exactly.
Issues that arise due to inaccuracies in the computation are discussed
in the next section.

\begin{figure}[t]
\noindent\hrulefill
\begin{mylist}{8mm}
\item[1.]
Let $N = \dim(\X\otimes\W\otimes\Y)$ and $M = \dim(\W)$, and define
\[
W_0 = \I_{\X\otimes\W\otimes\Y},\quad\quad
\rho_0 = W_0/N, \quad\quad
Z_0 = \I_{\W} 
\quad\quad\text{and}\quad\quad
\xi_0 = Z_0/M.
\]
Also let 
\[
\gamma = \frac{4}{3},
\quad\quad
\varepsilon = \frac{1}{64},
\quad\quad
\delta = \frac{\varepsilon}{2\norm{Q^{-1}}}
\quad\quad\text{and}\quad\quad
T = \left\lceil\frac{4\log(N)}{\varepsilon^3\delta}\right\rceil.
\]
\item[2.]
Repeat for each $t = 0,\ldots,T-1$:

\begin{mylist}{8mm}
\item[(a)]
Let $\Pi_t$ be the projection onto the positive eigenspaces of
the operator
\[
\Phi(\rho_t) - \gamma \,\I_{\X} \otimes \xi_t,
\]
where $\Phi$ is defined from $Q$ as in \eqref{eq:definition-of-Phi},
and set $\beta_t = \ip{\Pi_t}{\Phi(\rho_t)}$.

\item[(b)]
If $\beta_t \leq \varepsilon$ then {\it accept}, else let
\[
W_{t+1} = \exp\(-\epsilon\delta\sum_{j = 0}^t
\Phi^{\ast}(\Pi_j/\beta_j)\), \quad\quad
\rho_{t+1} = W_{t+1}/\tr(W_{t+1}),
\]
and
\[
Z_{t+1} = \exp\(\varepsilon\delta\sum_{j = 0}^t
\tr_{\X}(\Pi_j/\beta_j)\),
\quad\quad
\xi_{t+1} = Z_{t+1}/\tr(Z_{t+1}).
\]
\end{mylist}

\item[3.]
If acceptance did not occur in step 2, then {\it reject}.
\end{mylist}
\noindent\hrulefill
\caption{An algorithm that {\it accepts} if the optimal value of the
  semidefinite program in Section~\ref{sec:sdp} is larger than 7/8,
  and {\it rejects} if the optimal value is smaller than 5/8.}
\label{fig:algorithm}
\end{figure}

Assume first that the algorithm accepts, and write
\[
\rho = \rho_t,\quad
\Pi = \Pi_t,\quad
\xi = \xi_t\quad
\text{and}\quad
\beta = \beta_t
\]
for $t\in\{0,\ldots,T-1\}$ corresponding to the iteration in which
acceptance occurs.
For the sake of clarity, let us note explicitly that
\[
\rho\in\density{\X\otimes\W\otimes\Y},\quad\quad
\Pi\in\pos{\X\otimes\W}
\quad\quad\text{and}\quad\quad
\xi\in\density{\W}.
\]
We wish to prove that the optimal value of our semidefinite program is
at least $7/8$, and we will do this by constructing a primal feasible
solution that achieves an objective value strictly larger than $5/8$.

By the definition of $\Pi$, it holds that
\begin{equation} \label{eq:accept1}
\Pi \Phi(\rho) \Pi 
\geq \Pi (\Phi(\rho) -  \gamma\,\I_{\X}\otimes\xi) \Pi 
\geq \Phi(\rho) - \gamma\,\I_{\X}\otimes\xi,
\end{equation}
and by Lemma~\ref{lemma:inequality2} (which is stated and proved
below) it holds that
\begin{equation} \label{eq:accept2}
2 \I_{\X} \otimes \tr_{\X}\(\Pi \Phi(\rho) \Pi\)
\geq \Pi \Phi(\rho) \Pi.
\end{equation}
Combining the equations \eqref{eq:accept1} and \eqref{eq:accept2} one
has
\begin{equation} \label{eq:primal-bound-1}
\Phi(\rho) \leq \I_{\X} \otimes \( \gamma\,\xi +
2 \tr_{\X}\(\Pi\Phi(\rho)\Pi\)\).
\end{equation}
It therefore holds that
\[
X = \frac{\rho}{\gamma + 2 \ip{\Pi}{\Phi(\rho)}}
     \quad\quad\text{and}\quad\quad
\sigma = 
\frac{\gamma \xi + 2 \tr_{\X}\(\Pi \Phi(\rho) \Pi\)}
     {\gamma + 2 \ip{\Pi}{\Phi(\rho)}}
\]
represent a feasible solution to the primal problem under
consideration, achieving the objective value
\[
\frac{1}{\gamma + 2\ip{\Pi}{\Phi(\rho)}} = \frac{1}{\gamma + 2\beta} 
\geq \frac{1}{\gamma + 2\varepsilon} > \frac{5}{8}
\]
as required.

Now assume that the algorithm rejects, and consider the operator
\[
Y = \frac{(1 + 2\varepsilon)}{T}\sum_{t=0}^{T-1} \Pi_t/\beta_t.
\]
We claim that $Y$ is dual feasible and achieves an objective value
that is strictly smaller than $7/8$.
This will imply that the optimal value of the semidefinite program is
at most $5/8$.

Let us first prove that $Y$ is dual feasible.
It is clear that $Y$ is positive semidefinite, so it suffices to
prove that $\Phi^{\ast}(Y) \geq \I_{\X\otimes\W\otimes\Y}$, or
equivalently that $\lambda_N(\Phi^{\ast}(Y))\geq 1$.
Observe, for each $t = 0,\ldots,T-1$, that
\begin{align*}
\tr(W_{t+1})
& = 
\tr\left[
  \exp\(-\varepsilon\delta\Phi^{\ast}(\Pi_0/\beta_0+\cdots+\Pi_t/\beta_t)\)
  \right]\\
& \leq
\tr\left[\exp\(-\varepsilon\delta\Phi^{\ast}(\Pi_0/\beta_0+\cdots
  +\Pi_{t-1}/\beta_{t-1})\)
  \exp\(-\varepsilon\delta\Phi^{\ast}(\Pi_t/\beta_t)\)
  \right]\\
& = \tr\left[W_t \exp\(-\varepsilon\delta\Phi^{\ast}(\Pi_t/\beta_t)\)
  \right]
\end{align*}
by the Golden--Thompson inequality.
As each $\Pi_t$ is a projection operator, we have
\[
  \norm{\Phi^{\ast}(\Pi_t)} 
  = \norm{Q^{-1/2} (\Pi_t\otimes\I_{\Y}) Q^{-1/2}}
  \leq \norm{Q^{-1/2}}^2 = \norm{Q^{-1}},
\]
where we have used the sub-multiplicativity of the spectral norm to
obtain the inequality.
Given that $\beta_t > \varepsilon$ in the case at hand, it follows
that $\norm{\delta\Phi^{\ast}(\Pi_t/\beta_t)} < 1$.
By Lemma~\ref{lemma:exp-inequalities} (also presented below) it
therefore follows that
\[
\exp\(-\varepsilon\delta\Phi^{\ast}(\Pi_t/\beta_t)\)
\leq \I - \varepsilon\delta\exp(-\varepsilon) \Phi^{\ast}(\Pi_t/\beta_t).
\]
As each $W_t$ is positive semidefinite, we obtain
\begin{equation} \label{eq:W-frac-inequality}
\tr(W_{t+1}) \leq \tr(W_t)
\(1 - \varepsilon\delta\exp(-\varepsilon)
\ip{\frac{W_t}{\tr(W_t)}}{\Phi^{\ast}(\Pi_t/\beta_t)}\).
\end{equation}
Substituting $\rho_t = W_t/\tr(W_t)$ yields
\begin{align*}
  \tr(W_{t+1}) 
  & \leq \tr(W_t) \(1 - \varepsilon\delta\exp(-\varepsilon)
  \ip{\rho_t}{\Phi^{\ast}(\Pi_t/\beta_t)}\)\\
  & = \tr(W_t) \(1 - \varepsilon\delta\exp(-\varepsilon)\)\\
  & \leq \tr(W_t)\exp(-\varepsilon\delta\exp(-\varepsilon)),
\end{align*}
where the equality follows from
$\ip{\rho_t}{\Phi^{\ast}(\Pi_t)}=\ip{\Phi(\rho_t)}{\Pi_t}=\beta_t$
and the last inequality follows from the fact that
$1 + z \leq \exp(z)$ for all real numbers $z$.
As $\tr(W_0) = N$, it follows that
\begin{equation} \label{eq:reject1}
  \tr(W_T) \leq \tr(W_0)\exp(-T \varepsilon\delta\exp(-\varepsilon))
= \exp(-T \varepsilon\delta\exp(-\varepsilon) + \log(N)).
\end{equation}
On the other hand, we have
\begin{equation} \label{eq:reject2}
  \tr(W_T) 
  = \tr\left[\exp\(-\varepsilon\delta\sum_{t=0}^{T-1}
    \Phi^{\ast}\(\Pi_t/\beta_t\)\)\right] 
  \geq
  \exp\(-\varepsilon\delta\lambda_N \(\Phi^{\ast} 
  \(\sum_{t=0}^{T-1} \Pi_t/\beta_t\)\)\).
\end{equation}
Combining \eqref{eq:reject1} and \eqref{eq:reject2}, we have
\[
\lambda_N\(\Phi^{\ast} \(
\sum_{t=0}^{T-1} \Pi_t/\beta_t
\)\) \geq T \exp(-\varepsilon) -
\frac{\log(N)}{\varepsilon\delta}.
\]
Using the inequality 
$\exp(-\varepsilon)-\varepsilon^2/4 > 1-\varepsilon$, and
substituting the value of $T$ specified by the algorithm, we have
\[
\lambda_N(\Phi^{\ast}(Y))
\geq (1 + 2\varepsilon)
\( \exp(-\varepsilon) - \frac{\log(N)}{T \varepsilon\delta} \)
> (1 + 2\varepsilon)(1 - \varepsilon) > 1
\]
as required.

Now it remains to establish an upper bound on the dual objective value
achieved by $Y$.
A similar method to the one used to prove the feasibility of $Y$ above
will provide a suitable bound.
We begin by observing, for each $t = 0,\ldots,T-1$, that
\begin{align*}
\tr(Z_{t+1})
& = \tr\left[\exp\(\varepsilon\delta \tr_{\X} \(
  \Pi_0/\beta_0 + \cdots + \Pi_t/\beta_t \)\)\right]\\
& \leq \tr\left[\exp\(\varepsilon\delta \tr_{\X} \(
  \Pi_0/\beta_0 + \cdots + \Pi_{t-1}/\beta_{t-1} \)\)
  \exp\( \varepsilon\delta \tr_{\X}(\Pi_t/\beta_t) \)
\right]\\
& = \tr\left[Z_t\exp\( \varepsilon\delta \tr_{\X}(\Pi_t/\beta_t) \)
\right].
\end{align*}
Given that
\[
\norm{\tr_{\X}(\Pi_t)} \leq
\norm{ \(\bra{0} \otimes \I_{\W}\)\Pi_t\(\ket{0} \otimes \I_{\W}\)}
+ \norm{ \(\bra{1} \otimes \I_{\W}\)\Pi_t\(\ket{1} \otimes \I_{\W}\)}
\leq 2,
\]
and using the fact that $\beta_t > \epsilon$ in the case at hand, 
it follows that $\norm{\delta \tr_\X (\Pi_t/\beta_t)} < 1$.
We now apply Lemma~\ref{lemma:exp-inequalities} to obtain
\[
\exp\( \varepsilon\delta \tr_{\X}(\Pi_t/\beta_t) \)
\leq \I + \varepsilon\delta \exp(\varepsilon) \tr_{\X}(\Pi_t/\beta_t).
\]
As each $Z_t$ is positive semidefinite it follows that
\begin{equation} \label{eq:Z-frac-inequality}
  \tr(Z_{t+1})  \leq
  \tr(Z_t)\(1+\varepsilon\delta\exp(\varepsilon)
    \ip{\frac{Z_t}{\tr(Z_t)}}{\tr_{\X}(\Pi_t/\beta_t)}\).
\end{equation}
Substituting $\xi_t = Z_t/\tr(Z_t)$ gives
\[
\tr(Z_{t+1}) \leq
\tr(Z_t)\(1+ \varepsilon\delta \exp(\varepsilon)
  \ip{\xi_t}{\tr_{\X}(\Pi_t/\beta_t)}\)
=  \tr(Z_t)\(1+ \varepsilon\delta\exp(\varepsilon)
\ip{\I_{\X}\otimes \xi_t}{\Pi_t/\beta_t}\).
\]
Now, as $\ip{\Phi(\rho_t)-\gamma\I_{\X}\otimes\xi_t}{\Pi_t}\geq 0$, we
may again use the fact that $1 + z \leq \exp(z)$ for all real numbers
$z$ to obtain
\begin{equation} \label{eq:other-Z-inequality}
\tr(Z_{t+1}) 
\leq
\tr(Z_t)\(1+ \frac{\varepsilon\delta\exp(\varepsilon)}{\gamma}
\ip{\Phi(\rho_t)}{\Pi_t/\beta_t}\)
\leq
\tr(Z_t) \exp\(\frac{\varepsilon\delta\exp(\varepsilon)}{\gamma}\).
\end{equation}
Consequently
\[
\tr(Z_T) \leq 
\tr(Z_0)\exp\(\frac{T \varepsilon\delta \exp(\varepsilon)}{\gamma}\)
= \exp\(\frac{T \varepsilon\delta \exp(\varepsilon)}{\gamma} + \log(M)\).
\]
On the other hand we have
\[
\tr(Z_T) =
\tr\left[
\exp\( \varepsilon\delta \sum_{t = 0}^{T-1}
\tr_{\X}(\Pi_t/\beta_t) \)\right]
\geq \exp\(\varepsilon\delta
\lambda_1\(\tr_{\X} \(\sum_{t = 0}^{T-1}\Pi_t/\beta_t\)\)\),
\]
and therefore
\[
\lambda_1 \( \tr_{\X}\(
\sum_{t = 0}^{T-1} \Pi_t/\beta_t
\) \) \leq \frac{T \exp(\varepsilon)}{\gamma} +
\frac{\log(M)}{\varepsilon\delta}.
\]
Given that $M<N$ it follows that
\[
\norm{\tr_{\X}(Y)} =
\lambda_1(\tr_{\X}(Y))
\leq (1 + 2\varepsilon)\(
\frac{\exp(\varepsilon)}{\gamma} + \frac{\log(M)}{T\varepsilon\delta}\)
< \frac{7}{8}.
\]
Thus, $Y$ is a dual feasible solution whose objective value is
smaller than $7/8$, and we conclude that the optimal value of our
semidefinite program is at most $5/8$ as required.

It remains to state and prove the two lemmas that were required in the
analysis above.
They are as follows.

\begin{lemma} \label{lemma:inequality2}
  Let $P\in\pos{\X\otimes\Z}$ be any positive semidefinite operator,
  and assume that $\dim(\X) = 2$.
  Then $P \leq 2 \I_{\X} \otimes \tr_{\X}(P)$.
\end{lemma}

\begin{proof}
  Let $\sigma_x$, $\sigma_y$ and $\sigma_z$ denote the Pauli operators
  on $\X$.
  In matrix form they are
  \[
  \sigma_x = 
  \begin{pmatrix} 
    0 & 1\\ 1 & 0
  \end{pmatrix},
  \quad\quad
  \sigma_y = 
  \begin{pmatrix} 
    0 & -i\\ i & 0
  \end{pmatrix}
  \quad\quad\text{and}\quad\quad
  \sigma_z = 
  \begin{pmatrix} 
    1 & 0\\ 0 & -1
  \end{pmatrix}.
  \]
  As each of these operators is Hermitian, we have that
  $(\sigma_x \otimes\I_{\Z}) P (\sigma_x \otimes\I_{\Z})$,
  $(\sigma_y \otimes\I_{\Z}) P (\sigma_y \otimes\I_{\Z})$
  and $(\sigma_z \otimes\I_{\Z}) P (\sigma_z \otimes\I_{\Z})$ are
  positive semidefinite.
  It therefore holds that
  \[
  2 \I_{\X} \otimes \tr_{\X}(P) 
  =
  P 
  + (\sigma_x \otimes\I_{\Z}) P (\sigma_x \otimes\I_{\Z})
  + (\sigma_y \otimes\I_{\Z}) P (\sigma_y \otimes\I_{\Z})
  + (\sigma_z \otimes\I_{\Z}) P (\sigma_z \otimes\I_{\Z})
  \geq P
  \]
  as required.
\end{proof}

\begin{lemma} \label{lemma:exp-inequalities}
  Let $P$ be an operator satisfying $0\leq P\leq \I$.  
  Then for every real number $\eta > 0$, the following two inequalities
  hold:
  \begin{align*}
    \exp(\eta P) & \leq \I + \eta \exp(\eta)P,\\
    \exp(-\eta P) & \leq \I - \eta \exp(-\eta)P.
  \end{align*}
\end{lemma}

\begin{proof}
  It is sufficient to prove the inequalities for $P$ replaced by a
  scalar $\lambda\in[0,1]$, for then the operator inequalities follow by
  considering a spectral decomposition of $P$.
  If $\lambda=0$ both inequalities are immediate, so let us assume
  $\lambda>0$.
  By the Mean Value Theorem there exists a value
  $\lambda_0\in(0,\lambda)$ such that
  \[
  \frac{\exp(\eta \lambda) - 1}{\lambda} = \eta \exp(\eta \lambda_0)
  \leq \eta \exp(\eta),
  \]
  from which the first inequality follows.
  Similarly, there exists a value $\lambda_0\in(0,\lambda)$ such that
  \[
  \frac{\exp(-\eta\lambda)-1}{\lambda} = -\eta\exp(-\eta\lambda_0)
  \leq -\eta \exp(-\eta),
  \]
  which yields the second inequality.
\end{proof}

%-----------------------------------------------------------------------------%
\section{Proof that QIP is contained in PSPACE} \label{sec:proof}
%-----------------------------------------------------------------------------%

With the algorithm from the previous section in hand, the proof that
$\class{QIP}\subseteq\class{PSPACE}$ follows the same approach used in
\cite{JainUW09} to prove $\class{QIP}(2)\subseteq\class{PSPACE}$.
The proof is described in the two subsections that follow.

%-----------------------------------------------------------------------------%
\subsection{Simulation by bounded-depth Boolean circuits}
%-----------------------------------------------------------------------------%

Let $A = (A_{\yes},A_{\no})$ be a promise problem in $\class{QIP}$.
Our goal is to prove that $A\in\class{PSPACE}$.
Given that $\class{PSPACE} = \class{NC}(\mathit{poly})$, as was
mentioned in Section~\ref{sec:NC}, it suffices to prove
$A\in\class{NC}(\mathit{poly})$.

Using Theorem 5.4 of \cite{MarriottW05} we have that there exists a
single-coin $\class{QMAM}$-protocol for $A$ with perfect completeness
and soundness probability $1/2 + \varepsilon$, for
$\varepsilon = 1/64$.
(Of course any other sufficiently small positive constant would do,
and in fact one can replace $\varepsilon$ with an exponentially small
value---but this choice is sufficient for our needs.)
We will make a small modification in Arthur's specification so that he
always accepts outright with probability $4\varepsilon$, and otherwise
measures the registers sent by Merlin according to his original
specification.
With this modification in place, we have that if $x\in A_{\yes}$, then
Arthur can be made to accept with certainty, while if $x\in A_{\no}$
then the maximum probability with which Arthur can be made to accept
is smaller than $1/2 + 3\varepsilon$.
It also holds that every strategy of Merlin causes Arthur to accept
with probability at least $4\varepsilon$.

Now, for any fixed choice of an input string $x\in A_{\yes}\cup
A_{\no}$, let $Q$ be the operator defined from this modified
specification of Arthur on the input $x$ as was described in
Section~\ref{sec:sdp}.
Give that Arthur always accepts with probability at least
$4\varepsilon$, it follows that the smallest eigenvalue of $Q$ is at
least $2\epsilon$.
Therefore, $Q$ is invertible and satisfies $\snorm{Q^{-1}} \leq
1/(2 \varepsilon)$. 
Moreover, the semidefinite program defined by $Q$, as  described in
Section~\ref{sec:sdp}, has an optimal value that is
equal to 1 when $x\in A_{\yes}$ and smaller than $1/2 + 3\varepsilon$
when $x\in A_{\no}$.

Next, consider a two-step computation as follows:
\begin{mylist}{\parindent}
\item[1.]
  Compute from a given input string $x$ an explicit description of the
  operator $Q$ specified above.

\item[2.]
  Run an $\class{NC}$ implementation of the algorithm from
  Section~\ref{sec:algorithm} on $Q$.
\end{mylist}
The first step of this computation can be performed in
$\class{NC}(\mathit{poly})$ using an exact computation.
This follows from the fact that in $\class{NC}(\mathit{poly})$ one can
first compute explicit matrix representations of all of the gates in
the quantum circuit specifying Arthur's measurements, and then process
these matrices using elementary matrix operations to obtain $Q$.
Note that, without loss of generality, the description of $Q$ has
length polynomial in $N$, which (as defined in the algorithm) is the
dimension of the space on which it acts.

The second step of the computation, which is an $\class{NC}$
implementation of the algorithm from Section~\ref{sec:algorithm}, is
not quite as straightforward as the first step.
In fact, it is only possible for us to {\it approximate} this
algorithm in $\class{NC}$, as we only know how to approximate the
operator $Q^{-1/2}$, the matrix exponentials, and the spectral
decompositions needed to obtain the projection operators
$\Pi_0,\ldots,\Pi_{T-1}$.
Nevertheless, we claim that such an approximation is possible in
$\class{NC}$, with sufficient accuracy to distinguish the two cases
$x\in A_{\yes}$ and $x\in A_{\no}$.
This fact is argued in the subsection following this one.

Under the assumption that the second step is performed in
$\class{NC}$, we have that the composition of the two steps is an
$\class{NC}(\mathit{poly})$ computation.
We therefore obtain that $A\in\class{NC}(\mathit{poly})$ as required.

%-----------------------------------------------------------------------------%
\subsection{A high precision NC implementation of the algorithm}
%-----------------------------------------------------------------------------%

It remains to argue that the algorithm from
Section~\ref{sec:algorithm} can be approximated by an $\class{NC}$
computation with sufficient accuracy to distinguish the cases
$x\in A_{\yes}$ and $x\in A_{\no}$ as described above.
It will be evident from the discussion that follows that obtaining
sufficient accuracy in $\class{NC}$ is not a significant
challenge; and one could, in fact, demand much greater accuracy (by an
order of magnitude) and still be able to perform the computation in  
$\class{NC}$.

The first step in the implementation of the algorithm is to
approximate $Q^{-1/2}$.
In more precise terms, we first compute an operator $R$ such that
$R^2$ is a close approximation to $Q$, and then compute $R^{-1}$ in
$\class{NC}$ using an exact computation.
To compute $R$, we may compute a spectral decomposition of $Q$, and then
take $R$ to be the operator that results by replacing each eigenvalue
in this decomposition with its square root.
It is straightforward to perform high-precision approximations of
these computations in $\class{NC}$ with sufficient accuracy so that
$\norm{Q - R^2} \leq \varepsilon$ and $\norm{R^{-1}} \leq 1/\varepsilon$.
Now, if we compare two semidefinite programs, one defined by $Q$ as
specified in Section~\ref{sec:sdp} and the other defined similarly
with $Q$ replaced by $R^2$, we find that the optimal values are
close.
More specifically, given that $\norm{Q - R^2} \leq \varepsilon$, the
optimal values of the two semidefinite programs can differ by at most
$2\varepsilon$.
Thus, the optimal value of the semidefinite program for $R^2$ is
at least $1 - 2\varepsilon > 7/8$ in case $x\in A_{\yes}$ and
at most $1/2 + 5\varepsilon < 5/8$ in case $x\in A_{\no}$.

In the interest of clarity, to avoid introducing a new variable $R$
into the analysis that follows, let us simply redefine $Q$ at this
point to be $R^2$.
Thus, $Q^{-1/2} = R^{-1}$ is known exactly by our implementation of
the algorithm and all of the requirements on $Q$ are in place---which
are that $\norm{Q^{-1/2}}\leq 1/\varepsilon = 64$ and the optimal
value of the semidefinite program in Section~\ref{sec:sdp} defined by
$Q$ is at least $7/8$ if $x\in A_{\yes}$ and at most $5/8$ if $x\in
A_{\no}$.

Next, let us focus on the projection operators
\begin{equation} \label{eq:Pi-operators}
\Pi_0,\ldots,\Pi_{T-1} \in \pos{\X\otimes\W}
\end{equation}
and the density operators
\begin{equation} \label{eq:rho-and-xi-operators}
\rho_0,\ldots,\rho_T\in\density{\X\otimes\W\otimes\Y}
\quad\quad\text{and}\quad\quad
\xi_0,\ldots,\xi_T\in\density{\W}
\end{equation}
that are to be computed in the course of the algorithm.
We will choose an integer $K$ that we take to represent the number of
bits of accuracy with which these operators are stored.
In more precise terms, the algorithm will store the real and imaginary
parts of each of the entries of the above operators
\eqref{eq:Pi-operators} and \eqref{eq:rho-and-xi-operators} as
integers divided by $2^K$.
It will suffice to take $K = c\lceil\log(N)\rceil$, for a suitable
choice of a constant $c$, although one could in fact afford to take
$K$ to be polynomial in $N$ rather than logarithmic.
As each entry of these operators has absolute value at most 1, 
the total number of bits needed to represent the entire
collection of operators is $O(T K N^2)$, which is polynomial in $N$.

In addition to the above operators, the algorithm will store the
scalar values $\beta_0,\ldots,\beta_{T-1}$.
These values do not need to be approximated; each value $\beta_t$ is
computed exactly as the rational number defined by the operators
$\rho_t$ and $\Pi_t$ stored by the algorithm.
We will not consider that the operators
$W_1,\ldots,W_T$ and $Z_1,\ldots,Z_T$ are stored by the algorithm at
all, as their only purpose in the computation is to specify the
density operators $\rho_1,\ldots,\rho_T$ and
$\xi_1,\ldots,\xi_T$.

We will also take $\mu$ to be a small constant, say $\mu = 2^{-10}$,
that will represent an error parameter for the computation.
Similar to the choice of $K$, we could afford to take $\mu$ to be
significantly smaller than this and still be able to perform the
computation in $\class{NC}$.

Now, consider the two steps (a) and (b) that are performed within each
iteration of the loop in step 2 of the algorithm.
We must approximate these steps, and we demand the following accuracy
requirements when doing this.
For step (a), we will require that the projection operator $\Pi_t$
computed by the algorithm satisfies the condition 
\begin{equation} \label{eq:primal-bound-2}
\Pi_t(\Phi(\rho_t) - \gamma \I_{\X}\otimes \xi_t)\Pi_t 
\geq P_t -\frac{\mu}{M}\I_{\X\otimes\W},
\end{equation}
where $P_t$ is defined as the positive part of
$\Phi(\rho_t) - \gamma\I_{\X}\otimes\xi_t$.
It is possible to perform such a computation in $\class{NC}$ by
setting the error parameter $\eta$ in an approximate spectral
decomposition computation of $\Phi(\rho_t) -
\gamma\I_{\X}\otimes\xi_t$ as $\eta = \mu/(2M)$, for instance.
Then, $\Pi_t$ is taken to be the appropriately defined projection
operator rounded to $K$ bits of accuracy.
For step (b), we will require that
\begin{equation} \label{eq:W-and-Z-bounds}
\norm{\rho_{t+1} - W_{t+1}/\tr(W_{t+1})} < \frac{\mu\delta}{N}
\quad\quad
\text{and}
\quad\quad
\norm{\xi_{t+1} - Z_{t+1}/\tr(Z_{t+1})} < \frac{\mu\delta}{M}.
\end{equation}
In these inequalities we do not consider that $W_{t+1}$ and $Z_{t+1}$
are stored by the algorithm, but rather we consider that they are
operators {\it defined} by the equations
\[
W_{t+1} = \exp\(-\epsilon\delta\sum_{j = 0}^t
\Phi^{\ast}(\Pi_j/\beta_j)\)
\quad\quad\text{and}\quad\quad
Z_{t+1} = \exp\(\varepsilon\delta\sum_{j = 0}^t
\tr_{\X}(\Pi_j/\beta_j)\),
\]
for the particular operators $\Pi_0/\beta_0,\ldots,\Pi_{t}/\beta_t$
that are stored by the algorithm.
The algorithm's {\it approximations} of $W_{t+1}$ and $Z_{t+1}$
determine the density operators $\rho_{t+1}$ and $\xi_{t+1}$.
As the matrix exponentials are to be computed for operators having
norm bounded by $T = O(\log N)$, it is clear that $\rho_{t+1}$ and
$\xi_{t+1}$ with the required properties can be computed in
$\class{NC}$.

Finally, we have that the total number of iterations in the algorithm
is $T = O(\log N)$.
Given that each of the iterations of the algorithm can be performed in
$\class{NC}$, and that the total number of bits that must be stored
from one iteration to the next is polynomial in $N$, we have that the
composition of these $T$ iterations can be performed in $\class{NC}$
as well.

It remains only to show that the approximations
\eqref{eq:primal-bound-2} and \eqref{eq:W-and-Z-bounds} are sufficient to
guarantee that the algorithm accepts or rejects correctly.
This analysis is done in almost exactly the same way as was presented
in Section~\ref{sec:algorithm}.
Even though the operators
\[
\rho_0,\ldots,\rho_{T-1},\quad
\xi_0,\ldots,\xi_{T-1},\quad
\text{and}\quad
\Pi_0/\beta_0,\ldots,\Pi_{T-1}/\beta_{T-1}
\]
do not necessarily satisfy the precise equations that were assumed in
Section~\ref{sec:algorithm}, they may nevertheless be used to
construct primal and dual solutions to the semidefinite program that
satisfy the required bounds.

In the case that the algorithm accepts, a consideration of the
operators $\rho = \rho_t$, $\Pi = \Pi_t$, and $\xi = \xi_t$ as before
allows for the construction of a primal feasible solution with a large
objective value.
In place of \eqref{eq:primal-bound-1}, we have
\[
\Phi(\rho) \leq \I_{\X} \otimes \(\gamma \xi + 2 \tr_{\X}(\Pi
\Phi(\rho) \Pi) + \frac{\mu}{M} \I_{\W}\),
\]
which allows for a lower bound of $1/(\gamma + 2\varepsilon + \mu)$
for the primal objective function.
For our choice $\mu = 2^{-10}$ of an error bound, this quantity is
still lower-bounded by $5/8$, which implies that the algorithm has
operated correctly in this case.

A similar analysis to the one before holds for the case of rejection
as well.
We consider the operators
\[
\Pi_0/\beta_0,\,\ldots,\Pi_{T-1}/\beta_{T-1}
\]
produced by the algorithm, and take
\[
Y = \frac{(1 + 2\varepsilon)(1 + 2 \mu)}{T}\sum_{t = 0}^{T-1}
\Pi_t/\beta_t.
\]
When proving the dual feasibility of $Y$ we are no longer free to
substitute $\rho_t = W_t/\tr(W_t)$, but instead we must introduce a
small error term due to the fact that $\rho_t$ is just an
approximation to $W_t/\tr(W_t)$.
By the first inequality of \eqref{eq:W-and-Z-bounds} above we may
conclude that
\[
\ip{\frac{W_t}{\tr(W_t)}}{\Phi^{\ast}(\Pi_t/\beta_t)} \geq 1 - \mu;
\]
and by substituting this into \eqref{eq:Z-frac-inequality} and
following a similar argument to the one from before we obtain
\[
\lambda_N(\Phi^{\ast}(Y)) \geq (1 + 2\varepsilon)(1 + 2\mu)
\((1-\mu)\exp(-\varepsilon) - \frac{\varepsilon^2}{4}\) > 1.
\]
Thus, dual feasibility holds for $Y$.
Along similar lines, by using \eqref{eq:primal-bound-2} and
\eqref{eq:W-and-Z-bounds}, one finds again that the dual objective
value achieved by $Y$ less than $7/8$, and therefore the algorithm
operates correctly in this case as well.

%-----------------------------------------------------------------------------%
\subsection*{Acknowledgments}
%-----------------------------------------------------------------------------%

We thank Xiaodi Wu for helpful discussions.
Rahul Jain's research is supported by the internal grants of the
Centre for Quantum Technologies, which is funded by the Singapore
Ministry of Education and the Singapore National Research Foundation.
Zhengfeng Ji's research at the Perimeter Institute is supported by the
Government of Canada through Industry Canada and by the Province of
Ontario through the Ministry of Research \& Innovation.
Sarvagya Upadhyay's research is supported in part by Canada's NSERC,
CIFAR, MITACS, QuantumWorks, Industry Canada, Ontario's Ministry of
Research and Innovation, and the U.S.~ARO.
John Watrous's research is supported by Canada's NSERC and CIFAR.

\bibliographystyle{alpha}
%\bibliography{QIP-equals-PSPACE}

\end{document}